\documentclass{amsart}

\usepackage{anysize}
\usepackage{amsmath}
\usepackage{latexsym}
\usepackage{amsfonts}
\usepackage{amssymb}
\usepackage{amsthm}
\usepackage[all]{xy}

\theoremstyle{plain}

\newtheorem{proposition}{Proposition}
\newtheorem{corollary}{Corollary}
\newtheorem{theorem}{Theorem}
\theoremstyle{definition}

\newtheorem{remark}{Remark}
\newtheorem{definition}{definition}

\renewcommand{\aa}{\mathcal{A}}

\newcommand{\cl}{\mathcal{L}}

\newcommand{\bb}{\mathcal{B}}
\newcommand{\ee}{\mathcal{E}}
\newcommand{\ff}{\mathcal{F}}
\newcommand{\hh}{\mathcal{H}}

\newcommand{\de}{{\rm d}}

\newcommand{\lh}{\mathcal{L}(\mathcal{H})} 
\newcommand{\trh}{\mathcal{I}_1 (\mathcal{H})} 
\newcommand{\trhs}{\mathcal{I}_1^S (\mathcal{H})} 
\newcommand{\HS}{\mathcal{I}_2(\mathcal{H})} 
\newcommand{\sh}{\mathcal{S(H)}} 
\newcommand{\shs}{\mathcal{S}^S (\mathcal{H})} 
\newcommand{\shuno}{\mathcal{S}^1 (\mathcal{H})} 

\newcommand{\no}[1]{\left\|#1\right\|} 
\newcommand{\tr}[1]{{\rm tr} \left[ #1\right]} 

\newcommand{\frecc}{\longrightarrow}

\newcommand{\scal}[2]{\left\langle #1 \, , \, #2 \right\rangle}

\newcommand{\R}{\mathbb{R}} 
\newcommand{\C}{\mathbb C} 
\newcommand{\N}{\mathbb N} 
\newcommand{\T}{\mathbb{T}} 
\newcommand{\PP}{\mathbb{P}} 

\newcommand{\lft}{\left(}
\newcommand{\rgt}{\right)}

\newcommand{\ldue}[1]{L^2 \left( #1 \right)}
\newcommand{\luno}[1]{L^1 \left( #1 \right)}
\newcommand{\bor}[1]{\bb (#1)}

\begin{document}


\title[Quantum homodyne tomography as an informationally complete POVM]{QUANTUM HOMODYNE TOMOGRAPHY AS AN INFORMATIONALLY COMPLETE 
POSITIVE OPERATOR VALUED MEASURE}

\author{Paolo Albini}

\address{Dipartimento di Fisica, Universit\`a di Genova, Via Dodecaneso 33,\\
Genova, 16146, Italy\\
and I.N.F.N., Sezione di Genova, Via Dodecaneso 33,\\
Genova, 16146, Italy\\
albini@fisica.unige.it}

\author{Ernesto De Vito}

\address{Dipartimento di Scienze per l'Architettura, Universit\`a di Genova, 
Stradone S. Agostino 37\\
Genova, 16123, Italy,\\
and I.N.F.N., Sezione di
Genova, Via Dodecaneso~33\\
Genova, 16146, Italy\\
devito@dima.unige.it}

\author{Alessandro Toigo}

\address{Dipartimento di Informatica, Universit\`a di Genova, Via Dodecaneso 35\\ Genova, 16146, Italy,\\
and I.N.F.N., Sezione di Genova, Via Dodecaneso 33\\
Genova, 16146, Italy\\
toigo@ge.infn.it}

\maketitle


\begin{abstract}
We define a positive operator valued measure $E$ on $ [0,2\pi]
  \times\mathbb R$ describing the
  measurement of randomly sampled quadratures in quantum homodyne
  tomography, and we study its probabilistic properties. Moreover, we give a
  mathematical analysis of the relation between  the description of a
  state in terms of $E$ and the description provided by its Wigner
  transform.
\end{abstract}

\keywords{Quantum homodyne tomography; positive operator valued measures; Wigner function.}



\section{Introduction}

Quantum homodyne tomography \cite{Art,DAr,Leo,VR} allows to determine the state of a
single mode radiation field by repeated measurements of the quadrature observables $X_\theta$,
the phases $\theta$ being chosen randomly in
$\T=[0,2\pi]$. 
This can be seen as a consequence of the fact \cite{Cas} that, 
for a large class of observables $O$, there exists an associated
function $f_O:\T\times\R\to\R$ such that 
\begin{equation}\label{*}
\tr{ O\rho } = \int_0^{2\pi} \left[\int_{\R} f_O (\theta,x) \de
  \nu_\theta^\rho (x) \right] \frac{\de \theta}{2\pi} , 
\end{equation}
where $\nu_\theta^\rho$ is the probability distribution of 
outcomes obtained for the quadrature $X_\theta$ measured on the state
$\rho$ and $\de \theta / 2\pi$ is the uniform probability
distribution on $\T$.
Actual reconstrunction schemes are strictly related to a statistical interpetation
of formulas of this kind.
Indeed, quantum tomography experiments output a $n$-uple 
$\{(\Theta_i, X_i)\}_{i = 1}^n$ of pairs in $\T \times \R$, each one of which represents
the outcome $X_i$ of a measurement of the quadrature observable corresponding to 
the randomly picked phase $\Theta_i$. If one assumes such pairs to be samples from a random 
variable on $\T \times \R$ distributed accordingly to a probability measure $\mu^\rho$ such that
\begin{equation}\label{spezz}
\de\mu^\rho(\theta, x)=\de\nu_\theta^\rho (x) \frac{\de \theta}{2\pi},
\end{equation}
one can use the experimental outcomes to estimate integrals such as 
\eqref{*} (see Ref.~\cite{DAr} and references therein), for example by replacing $\de \mu^\rho$ with its empirical estimate
$\frac{1}{n} \sum_{i=1}^n \delta_{(\Theta_i,X_i)}$.

The above reconstruction formula, 
although very popular, is not the only scheme used for tomographical state estimation: other ones are known which don't rely on it \cite{DAr}.
The hypotesis that experimental results are distributed accordingly to \eqref{spezz} lies however under both every proposed 
reconstruction algorithm and its statistical analysis \cite{Art,DAr,Leo,Butu,Guta}. Actually, although given for granted in the cited literature, well-definiteness of a joint probability 
distribution such as $\mu^\rho$ in \eqref{spezz} is {\em a priori} not trivial.
In the first part of our paper we prove well defineteness of $\mu^\rho$ by showing there exists a positive operator valued 
measure (POVM) $E$ on  $\T\times\R$ such that  
\[\mu^\rho(Z)=tr{[E(Z)\rho]}\qquad\text{for any Borel subset $Z$ of }\T\times\R.\] 
According to the physical meaning of $\mu^\rho$, the POVM $E$ is the
{\em generalized observable} associated with the quantum homodyne
  tomography experimental setup. In particular, we show that $\mu^\rho$ has density
  $p^\rho(\theta,x)$ with respect to the 
  Lebesgue measure on  $\T\times \R$,  its
  support is always an unbounded set, and the mapping
  $\rho\mapsto \mu^\rho$ is injective (i.~e.~$E$ is {\em
    informationally complete}). 
The intertwining property
$ X_\theta= e^{i\theta N}X e^{-i\theta N}$, where $N$ is the number
operator and $X$ is the position operator, turns out to be crucial for the definition of $E$ (or, equivalently, for the definition of $\mu^\rho$).
We remark that the introduction of a POVM for the homodyne tomography measurement process is already 
present in physical literature (see section 2.3.2 in Ref.~ \cite{DAr}), but it is grounded on a formal construction.
We provide here an alternative, rigorous formulation.

In their seminal paper on quantum homodyne tomography \cite{VR}, Vogel and Risken argued that the the Radon transform of
$W(\rho)$, where $W(\rho)$ is the Wigner function associated to $\rho$, is precisely the probability density function $p^\rho$ generated by
the homodyne tomography measurement,  so that the following commutative diagram holds:
\[
\xymatrix{ & \text{States on }\mathcal H \ar_W[dl]\ar^{p^{\cdot}}[dr] & \\
  \text{Wigner functions on }\R^2
  \ar@<0.5ex>[rr]^{\text{Radon transform}} & & \text{probability densities on }\T\times\R } 
\] 
The suggested estimation procedure, applied also in the first homodyne tomography
experiments, is then based on the inversion of the Radon transform by means of classical
techniques in medical tomography. However, the derivation of this
fact is once again rather formal, and never given a rigorous basis in the literature on the subject:
in the second part of the paper, we will thus address the problems that arise in looking at such formulation 
of quantum tomography from a rigorous point of view. First of all, we will recall that in order for the Radon
transform to be well-defined, we need the Wigner function $W(\rho)$ to be integrable on $\R^2$. 
Then we will show that the support of $W(\rho)$ 
can {\em never} be bounded. This is a potential problem, since the estimation techniques used in classical tomography are 
explicitly devised for compactly supported objects. One can however by-pass the problem and still
give an inverse for the Radon transform if he assumes that the Wigner function under observation
is a Schwartz function on $\R^2$. This is precisely what happens in most homodyne tomography experiments,
where the states under observation are linear combinations of coherent or number states.
In Section \ref{secRadon} we show that this assumption on the Wigner
function is equivalent to suppose that $\rho$ has a kernel which is a
Schwartz function on $\R^2$ (since $\rho$ is an Hilbert-Schmidt
operator, $\rho$ is an integral operator whose kernel is a function on
$\R^2$). Under this assumption on $\rho$ we prove that the Radon
transform of $W(\rho)$ is $p^\rho$ and the inversion formula holds true. 

\section{Preliminaries and notations}\label{preliminari}
In this section, we will introduce the notations and give
a very brief description of the mathematical structure of quantum
homodyne tomography.

\subsection{Notations}
Let $\hh$ be a complex, separable Hilbert space with norm $\no{\cdot}$
and scalar product $\scal{\cdot}{\cdot}$ linear in the second
entry. Denote by $\lh$ the Banach space of the bounded operators on
$\hh$ with uniform norm $\no{\cdot}_{\cl}$. Let $\trh$ be the Banach
space of the trace class operators on $\hh$ with trace class norm
$\no{\cdot}_1$, and let $\sh$ be the convex subset of positive trace
one elements in $\trh$. Finally, let $\HS$ be the
Hilbert-Schmidt operators on $\hh$, with norm $\no{A}_2 =
\left[\tr{A^\ast A}\right]^{1/2}$. We recall that the elements of
$\sh$ are the {\em states} of the quantum system whose associated
Hilbert space is $\hh$.

Suppose $\Omega$ is a Hausdorff locally compact second countable
topological space. Let $\bor{\Omega}$ be the Borel $\sigma$-algebra of
$\Omega$.
We recall the following definition of positive operator valued measure.
\begin{definition}\label{defPOVM1}
A {\em positive operator valued measure} (POVM) on $\Omega$ with
values in $\hh$ is a map $E: \bor{\Omega} \frecc \lh$ such that 
\begin{itemize}
\item[({\rm i})] $E(A) \geq 0$ for all $A \in\bor{\Omega}$;
\item[({\rm ii})] $E(\Omega) = I$;
\item[({\rm iii})] if $\{ A_i \}_{i\in I}$ is a denumerable sequence of pairwise
  disjoint sets in $\bor{\Omega}$, then 
$$
E (\cup_i A_i) = \sum\nolimits_i E(A_i) ,
$$
where the sum converges in the weak (or, equivalently, ultraweak or
strong) topology of $\lh$. 
\end{itemize}
$E$ is a {\em projection valued measure} (PVM) if $E(A)^2 = E(A)$ for
all $A\in\bor{\Omega}$. 
\end{definition}
If $E$ is a POVM and $T\in\trh$, we define
$$
\mu^T_E (A) = \tr{E(A) T} \quad \forall A\in\bor{\Omega} .
$$
Then, $\mu^T_E$ is a bounded complex measure on $\Omega$. If $\rho\in\sh$, $\mu^\rho_E$ is actually a
probability measure on $\Omega$, and $\mu^\rho_E (A)$ is the
probability of obtaining a 
result in $A$ when performing a measurement of $E$ on the state
$\rho$. 
 
\subsection{The mathematics of quantum homodyne tomography}

The physical system of quantum homodyne tomography is a single radiation mode of the electromagnetic field.
The associated Hilbert space is $\hh =
\ldue{\R}$. Let 
$$
\aa = \left\{ p(x) \, e^{-\frac{x^2}{2}} \mid p \textrm{ is a polinomial}, \right\}
$$
which is a dense subspace of $\hh$. As usual, we denote by $X$ and $P$ the position and momentum
operators, respectively.  Their action on $\aa$ is explicitly given by
\[ (Xf)(x)=x f(x) \qquad\text{and}\qquad (Pf)(x)=-i\frac{\de f}{\de
  x}(x).\]
Letting $\T=[0,2\pi]$, for
any $\theta\in\T$ the corresponding quadrature is the
self-adjoint operator $X_\theta$ on $\ldue{\R}$,  whose action on $\aa$ is
\[ X_\theta  =\cos{\theta} X + \sin{\theta} P.\]
If $x,y\in\R$, and $x = r\cos \theta$, $y = r\sin\theta$, we have
$$
\left[ e^{ir X_\theta} f \right] (z) = \left[ e^{i (xX + yP)} f \right] (z) = e^{i\left( \frac{xy}{2} + xz \right)} f(z+y)
$$
for all $f\in \ldue{\R}$.

We denote by $\Pi_\theta$ the PVM on $\R$ associated to $X_\theta$ by spectral
theorem. In particular, $\Pi (A) := \Pi_0 (A)$ is just
multiplication in $\ldue{\R}$ 
by the characteristic function $1_A$ of $A$, while $\Pi_{\frac{\pi}{2}} (A) =
\ff^\ast \Pi (A) \ff$, where $\ff$ is the Fourier transform 
\begin{equation}\label{FT}
\ff f = \frac{1}{\sqrt{2\pi}} \int_{\R} e^{-ixy} f(y) \de y \quad f\in L^1\cap L^2 (\R) .
\end{equation}

The number operator is the essentially self-adjoint operator $N$
whose action on $\aa$ is 
\[N = \frac{1}{2} \lft X^2 + P^2 - 1 \rgt .\]
For all $\theta\in\T$, we let $V(\theta)=e^{i\theta N}$. Since the spectrum of $N$ is $\N$, the map $\theta\to V(\theta)$ is a unitary continuous
representation of $\T$ acting on $\ldue{\R}$, where we regard $\T$ as
a topological abelian group with addition modulo $2\pi$. The number representation $V$ intertwines the quadratures $X_\theta$, in the sense that
\begin{equation*}
X_\theta = V(\theta) X V(\theta)^\ast
\end{equation*}
for all $\theta\in\T$, and
\begin{equation*}
\Pi_\theta (A) = V(\theta) \Pi (A) V(\theta)^\ast
\end{equation*}
for all $\theta\in\T$ and $A\in\bor{\R}$.

Finally, given $\rho\in\sh$ and $\theta\in\T$, we denote by
$\nu^\rho_\theta$ the probability distribution on $\R$ of the outcomes
of the quadrature $X_\theta$ measured on the state $\rho$, namely
\begin{equation}
  \label{nurt}
   \nu^\rho_\theta (A) = \tr{\rho \Pi_\theta (A)} = \tr{\rho V(\theta)
  \Pi (A) V(\theta)^\ast} \quad \forall A \in \bor{\R} .
\end{equation}

\section{Main results}\label{mainresults}

In this section, we will describe explicitly the POVM which
intervenes in homodyne tomography and the associated probability
distributions on states. 

The first result studies some properties of the family of probability
measures $\nu^\rho_\theta$ defined by \eqref{nurt}. 
In its proof and in the statement of some of the following results, we will make use of the concept of
section through some $\theta \in \T$ of a Borel set $B \in \bor{\T \times \R}$, defined as
follows:
$$
B^\theta = \{x \in \R \vert (\theta, x) \in B\}.
$$
\begin{proposition}
Given $\rho\in\sh$ and $\theta\in\T$
\begin{itemize}
\item[({\rm i})] the probability measure
$\nu^\rho_\theta$ has density $p_\theta^\rho \in \luno{\R}$ with
respect to the Lebesgue measure on  $\R$;
\item[({\rm ii})] the map $\theta \mapsto \nu^\rho_\theta (B^\theta)$ is
  measurable for any $B\in\bor{\T\times \R}$.
\end{itemize}
\end{proposition}
\begin{proof}
If $A\in \bor{\R}$ has zero Lebesgue measure, then $\Pi (A) f = 1_A
f = 0$ for all $f\in \ldue{\R}$. Therefore, $\nu^\rho_\theta (A) =
\tr{\rho V(\theta) \Pi (A) V(\theta)^\ast} = 0$. Thus, the first claim
follows.

If $\{ e_n \}_{n\in\N}$ is a Hilbert basis of $\hh$, then
\begin{eqnarray*}
\tr{\rho V(\theta) \Pi (B^\theta) V(\theta)^\ast} & = & \sum_n
\scal{e_n}{V(\theta)^\ast \rho V(\theta) \Pi (B^\theta) e_n} \\ 
& = & \sum_n \sum_m \scal{e_n}{V(\theta)^\ast \rho V(\theta) e_m}
\scal{e_m}{\Pi (B^\theta) e_n} . 
\end{eqnarray*}
Since the map $\theta \mapsto \scal{e_n}{V(\theta)^\ast \rho V(\theta)
  e_m}$ is continuous and the map $\theta \mapsto \scal{e_m}{\Pi
  (B^\theta) e_n} = \int 1_B (\theta,x) e_n (x) \overline{e_m (x)}
\de x$ is measurable by Fubini theorem, measurability of $\theta
\mapsto \tr{\rho V(\theta) \Pi  (B^\theta) V(\theta)^\ast}$ follows.
\end{proof}

Next theorem shows the existence of a POVM associated to
quantum homodyne tomography. This theorem should be compared with the
formal derivation of $E$ given in Ref.~ \cite{DAr} (see eq.~(2.34) therein). 
\begin{theorem}\label{PropPOVM}
There exists a unique positive operator valued measure $E$ on $\T
\times \R$ acting in $\ldue{\R}$ such that
\begin{equation}
  \label{rabbia}
  \tr{\rho E(B)}= \int_{\T} \nu_\theta^\rho (B^\theta) \frac{\de\theta}{2\pi}
\end{equation}
for all $\rho\in\sh$ and $B\in\bor{\T \times \R}$.
\end{theorem}
\begin{proof}
Eq.~\eqref{nurt} suggests to define the POVM as
$$
E(B) = \int_{\T} V(\theta) \Pi (B^\theta) V(\theta)^\ast \frac{\de\theta}{2\pi} .
$$
To prove that the above definition is correct, we first show that 
the map $\theta \mapsto V(\theta) \Pi (B^\theta) V(\theta)^\ast$ is
$\frac{\de\theta}{2\pi}$-ultraweakly integrable for all Borel subsets $B$ of
$\T\times\R$, and then we prove that $B\mapsto E(B)$ is a POVM.

Now, given $B\in\bor{\T \times \R}$ and $\rho\in\sh$,
the map $\theta \mapsto \tr{\rho V(\theta) \Pi (B^\theta)
  V(\theta)^\ast}$ is measurable by the previous proposition, and 
$$
\left| \tr{\rho V(\theta) \Pi (B^\theta) V(\theta)^\ast} \right| \leq
\no{\rho}_1 \no{\Pi (B^\theta)}_{\cl} \leq 1 \quad \forall \theta \in
\T . 
$$
Therefore, it is $\frac{\de\theta}{2\pi}$-integrable. This shows that
$\theta \mapsto V(\theta) \Pi (B^\theta) V(\theta)^\ast$ is
$\frac{\de\theta}{2\pi}$-ultraweakly integrable. 
 
Suppose $T\in\trh$. Then $T = \sum_{k=0}^3 i^k T_k$, with $T_k \geq 0$
and $\no{T_0}_1 + \no{T_2}_1 = \no{T_1}_1 + \no{T_3}_1 \leq
\no{T}_1$. Setting $\rho_k = T_k / \no{T_k}_1$ (with $0/0 = 0$), we
see that 
\begin{eqnarray*}
\left| \int_{\T} \tr{T V(\theta) \Pi (B^\theta) V(\theta)^\ast}
  \frac{\de\theta}{2\pi} \right| & \leq & \sum_{k=0}^3 \no{T_k}_1
\int_{\T} \left| \tr{\rho_k V(\theta) \Pi (B^\theta) V(\theta)^\ast}
\right| \frac{\de\theta}{2\pi} \\ 
& \leq & \sum_{k=0}^3 \no{T_k}_1 \leq 2 \no{T}_1 . 
\end{eqnarray*}
This shows the existence of $E(B)\in\lh$. Clearly, $E(B) \geq 0$, and
$E(\T \times \R) = I$. 

If $\{ B_n \}_{n\in\N}$ is a monotone increasing family of elements in
$\bor{\T \times \R}$, with $B_n \uparrow B$, then, for all $\theta$, 
$$
\tr{\rho V(\theta) \Pi (B_n^\theta) V(\theta)^\ast} = \nu^\rho_\theta
(B_n^\theta) \uparrow \nu^\rho_\theta (B^\theta) = \tr{\rho V(\theta)
  \Pi (B^\theta) V(\theta)^\ast} . 
$$
By dominated convergence theorem
$$
\int_{\T} \tr{\rho V(\theta) \Pi (B_n^\theta) V(\theta)^\ast}
\frac{\de\theta}{2\pi} \uparrow \int_{\T} \tr{\rho V(\theta) \Pi
  (B^\theta) V(\theta)^\ast} \frac{\de\theta}{2\pi} , 
$$
and ultraweak $\sigma$-additivity of $E$ follows.
\end{proof}

We let $\mu^\rho = \tr{E(\cdot) \rho}$ be the probability distribution
on $\T\times \R$ associated to a measurement of $E$ performed on the
state $\rho$. By definition \eqref{rabbia} it follows that
\begin{equation}\label{murho}
\mu^\rho (B) = \int_{\T} \nu_\theta^\rho (B^\theta) \frac{\de\theta}{2\pi} 
\end{equation}
as wanted.
The following theorem gives some properties of $\mu^\rho$.
\begin{theorem}\label{TeoMis}
Let $\rho\in\sh$.
\begin{itemize}
\item[({\rm i})] The measure $\mu^\rho$ has density with respect to
  $\frac{\de\theta}{2\pi} \, \de x$. We denote such density by
  $p^\rho$. 
\item[({\rm ii})] For $\frac{\de\theta}{2\pi}$-almost all $\theta$, $p^\rho
  (\theta,x) = p^\rho_\theta (x)$ for $\de x$-almost all $x$. 
\item[({\rm iii})] The marginal probability distribution induced by $\mu^\rho$ on
  $\T$ is the Haar measure $\frac{\de\theta}{2\pi}$, and the
  conditional probability distribution induced by $\mu^\rho$ on $\R$
  is $\nu^\rho_\theta$ for $\frac{\de\theta}{2\pi}$-almost all
  $\theta$. 
\end{itemize}
\end{theorem}
\begin{proof}
\begin{itemize}
\item[({\rm i})] If $B\in\bor{\T \times \R}$ is a $\frac{\de\theta}{2\pi} \, \de
  x$-null set, then $B^\theta$ is $\de x$-null for
  $\frac{\de\theta}{2\pi}$-almost all $\theta$ by Fubini theorem, so,
  for such $\theta$'s, $\nu^\rho_\theta (B^\theta) = 0$. Therefore,
  $\mu^\rho (B) = 0$ by \eqref{murho}, thus showing that $\mu^\rho$
  has density with respect to $\frac{\de\theta}{2\pi} \, \de x$. 
\item[({\rm ii})] If $Z\in\bor{\T}$, $A\in\bor{\R}$, we have
$$
\int_Z \frac{\de\theta}{2\pi} \int_A p^\rho (\theta,x) \de x =
\mu^\rho (Z\times A) = \int_Z \nu^\rho_\theta (A)
\frac{\de\theta}{2\pi} . 
$$
This holds for all $Z$, implying that there exists a
$\frac{\de\theta}{2\pi}$-null set $N_A \in \bor{\T}$ such that $p^\rho
(\theta,\cdot)$ is $\de x$-integrable with 
$$
\int_A p^\rho (\theta,x) \de x = \nu^\rho_\theta (A)
$$
for all $\theta\notin N_A$.

Let $\{ A_n \}_{n\in\N}$ be a sequence in $\bor{\R}$ with the
following property: if $\mu_1 , \, \mu_2$ are positive measures on $\R$ such that
$\mu_1 (A_n) = \mu_2 (A_n)$ for all $n$, then $\mu_1 = \mu_2$ (such
sequence exists since $\R$ is second countable by Theorems C \S 5 and
A \S 13 in Ref.~ \cite{Hal}). Let $N = \cup_n N_{A_n}$. Then $N$ is
$\frac{\de\theta}{2\pi}$-null, and, if $\theta\notin N$, $p^\rho
(\theta,\cdot)$ is integrable with  
$$
\int_{A_n} p^\rho (\theta,x) \de x = \int_{A_n} p^\rho_\theta (x) \de x \quad \forall n .
$$
This implies that, if $\theta\notin N$, $p^\rho (\theta,x) = p^\rho_\theta
(x)$ for $\de x$-almost all $x$. 
\item[({\rm iii})] This is just \eqref{murho}.
\end{itemize}
\end{proof}
\begin{remark}
As a consequence of item~({\rm iii}) in the above proposition, a well known result on conditional
probability distribution ensures that, if $\phi$ is a $\mu^\rho$-integrable function, then
 $\phi(\theta,\cdot)$ is $\nu^\rho_\theta$-integrable for
 $\frac{\de\theta}{2\pi}$-almost all $\theta$, the map $\theta\mapsto
 \int_\R \phi(\theta,x) \de \nu^\rho_\theta (x)$ is
 $\frac{\de\theta}{2\pi}$-integrable, and 
\[
\int_{\T\times \R} \phi(\theta,x) \de \mu^\rho (\theta,x) = \int_\T \left[ \int_\R
  \phi(\theta,x) \de\nu^\rho_\theta (x) \right] \frac{\de\theta}{2\pi}. 
\]
\end{remark}

By Theorem \ref{TeoMis}, $E$ is the POVM associated to the
measurement of a quadrature $X_\theta$ chosen randomly from $\T$ with
uniform probability $\frac{\de\theta}{2\pi}$. 

The next corollary shows that the probability distribution $\mu^\rho$
can not have compact support for any $\rho\in\sh$. 
\begin{corollary}\label{suppprho}
For all $R>0$, we have
$$
\int_{\T} \int_{|x| > R} p^\rho (\theta,x) \de x \frac{\de\theta}{2\pi} > 0 .
$$
\end{corollary}
\begin{proof}
With $A_R = \{ x\in\R \mid |x|>R \}$, we have
\begin{eqnarray*}
\int_{\T} \int_{|x| > R} p^\rho (\theta,x) \de x
\frac{\de\theta}{2\pi} & = & \mu^\rho (\T \times A_R) = 
\int_{\T} \tr{\rho V(\theta) \Pi (A_R) V(\theta)^\ast} \frac{\de\theta}{2\pi} \\
& = & \tr{\rho^\prime \Pi (A_R)} ,
\end{eqnarray*}
with $\rho^\prime = \int_{\T} V(\theta)^\ast \rho V(\theta)
\frac{\de\theta}{2\pi}$. $\rho^\prime$ is a trace one positive
operator. Since it commutes with the representation $V$ of $\T$, it is
diagonal in the number basis $\{e_n\}_{n\in\N}$ of $\ldue{\R}$. Since
$\scal{e_n}{\Pi (A_R) e_n} > 0$ for all $n$, the claim follows. 
\end{proof}

As a consequence, the map $\rho\mapsto p^\rho$ from $\sh$ to the set
$P(\T\times \R)$ of probability densities in $\luno{\T\times \R}$ is
not surjective. The next corollary shows that it is actually injective, i.~e.~the POVM $E$ is {\em
  informationally complete} \cite{OQP}. 
\begin{corollary}
If $\rho, \sigma\in\sh$ and $\rho \neq \sigma$, then $\mu^\rho \neq \mu^\sigma$.
\end{corollary}
\begin{proof}
If $\rho, \sigma\in\sh$, then $\mu^\rho = \mu^\sigma$ if and only if
$p^\rho = p^\sigma$ (in $\luno{\T\times \R}$), which amounts to say
that $p^\rho_\theta = p^\sigma_\theta$ (in $\luno{\R}$) for
$\frac{\de\theta}{2\pi}$-almost all $\theta$. This is in turn
equivalent to $\nu^\rho_\theta = \nu^\sigma_\theta$ for
$\frac{\de\theta}{2\pi}$-almost all $\theta$. For $r\in\R$ and $\theta \in\T$, we have by spectral theorem 
$$
\int_{\R} e^{irx} \de \nu^\rho_\theta (x) = \int_{\R} e^{irx} \tr{\rho
  \Pi_\theta (\de x)} = \tr{\rho e^{ir X_\theta}} =
\sqrt{2\pi} \, [V(\rho)](r\cos \theta , r\sin\theta) , 
$$
where
$$
[V(\rho)] (x,y) = \frac{1}{\sqrt{2\pi}} \tr{\rho e^{i(xX + yP)}} .
$$
Since the map $V:\trh\frecc C(\R^2)$ is injective (see for example
Ref.~ \cite{Fol}), injectivity of the map $\rho\mapsto \mu^\rho$ follows. 
\end{proof}

\section{The Radon transform of the Wigner function and Radon
  reconstruction formula}\label{secRadon} 

In the previous section, by means of the POVM $E$ defined in Theorem
\ref{PropPOVM} we estabilished a convex injective correspondence
$\rho\mapsto p^\rho$ between states and the set of
probability densities on $\T\times\R$. However, no explicit
formula relating $\rho$ to the function $p^\rho$ was given, due to the
fact that, if $\rho$ does not have a simple expression in terms of the
number basis, the expression $\tr{V(\theta)^\ast \rho V(\theta) \Pi
  (A)}$ can not be explicitly computed.

In this section, we will show that, if the state $\rho$ is
sufficiently regular, $p^\rho$ can indeed be evaluated, being in fact
the Radon transform of the Wigner function $W(\rho)$ of $\rho$. This
is a very well known fact in quantum tomography, going back to the
seminal paper of Vogel and Risken \cite{VR}. However, no attention has never been paid in
the literature to the fact that performing the Radon transform of
$W(\rho)$ makes sense only for a restricted class of states, namely
for those $\rho\in\sh$ such that $W(\rho)\in\luno{\R^2}$. This
constraint becomes even more stringent when one considers the inverse
formula reconstructing $\rho$ (or, better, $W(\rho)$) from its
associated probability density $p^\rho$. We will see that, in order to
derive mathematically consistent formulas both for $p^\rho$ and the
reconstruction of $W(\rho)$, one needs to assume that the state
belongs to the set of  Schwartz functions on $\R^2$. This seems a
rather strong limitation, as the very natural attempt to extend the Radon
transform and Radon reconstruction to the whole set $\sh$ by means of
distribution theory fails in the quantum context (Remark
\ref{RemDistr}). Our main reference to the results below is Ref.~ \cite{Hel}.

If $T\in \trh$, we introduce the bounded continuous function $V(T)$ on $\R^2$, given by
\begin{equation}\label{WW}
[V(T)](x,y) = \frac{1}{\sqrt{2\pi}} \, \tr{T e^{i(xX+yP)}} .
\end{equation}
It is well known (see for example Ref.~ \cite{Fol}) that $V(T) \in
\ldue{\R^2}$, and $V$ uniquely extends to a unitary operator $V : \HS
\frecc \ldue{\R^2}$. The {\em Wigner transform} of $A\in \HS$ is just
(up to a constant) the Fourier transform of $V(A)$, i.~e. 
\begin{equation}\label{W}
W(A) = \frac{1}{\sqrt{2\pi}} \, \ff_2 V(A) ,
\end{equation}
where $\ff_2 = \ff \otimes \ff$ on $\ldue{\R^2} = \ldue{\R} \otimes
\ldue{\R}$, with $\ff$ defined in \eqref{FT}. 

If $f\in \luno{\R^2}$, the {\em Radon transform} of $f$ is the complex
function $Rf\in \luno{\T\times \R}$ given by 
\begin{equation}\label{TdR}
Rf (\theta , r) = \int_{-\infty}^{+\infty} f (r\cos \theta - t\sin \theta , r\sin
\theta + t\cos \theta) \de t 
\end{equation}
$\frac{\de\theta}{2\pi} \, \de r$-almost everywhere. 


We have the following fact.
\begin{proposition}\label{PropTdRdW}
If $W(\rho)\in\luno{\R^2}$, then
\begin{equation}\label{TdRdW}
[RW(\rho)] (\theta,r) = p^\rho (\theta,r) 
\end{equation}
for $\frac{\de\theta}{2\pi} \de r$-almost all $(\theta,r)$.
\end{proposition}
\begin{proof}
Let $\gamma : \T \times \R \frecc \R^2$ be the map
$$
\gamma (\theta,r) = \lft r\cos\theta , r\sin\theta \rgt .
$$
We have
$$
[V(\rho)\circ \gamma] (\theta,r) = \frac{1}{\sqrt{2\pi}} \tr{\rho e^{ir  {X}_\theta}} = \frac{1}{\sqrt{2\pi}} \int_{-\infty}^{+\infty} e^{irt} p^\rho_\theta (t) \de t = \left[ \ff^{-1} p^\rho_\theta \right] (r)
$$
by spectral theorem. On the other hand,
\begin{eqnarray*}
&& [\ff_2^{-1} W(\rho)\circ \gamma] (\theta,r) = \frac{1}{2\pi} \int_{-\infty}^{+\infty} \int_{-\infty}^{+\infty} e^{i (xr\cos\theta + yr\sin\theta)} [W(\rho)] (x,y) \de x \de y \\
&& \qquad \qquad = \frac{1}{2\pi} \int_0^{\pi} \int_{-\infty}^{+\infty} e^{i t r (\cos\phi\cos\theta + \sin\phi\sin\theta)} [W(\rho)] (t\cos\phi , t\sin\phi) |t| \de t \frac{\de \phi}{2\pi} \\
&& \qquad \qquad = \frac{1}{2\pi} \int_0^{\pi} \int_{-\infty}^{+\infty} e^{i tr \cos (\phi-\theta)} [W(\rho)] (t\cos\phi , t\sin\phi) |t| \de t \frac{\de \phi}{2\pi} \\
&& \qquad \qquad = \frac{1}{2\pi} \int_0^{\pi} \int_{-\infty}^{+\infty} e^{i tr \cos\phi} [W(\rho)] (t\cos(\phi + \theta) , t\sin(\phi + \theta)) |t| \de t \frac{\de \phi}{2\pi} \\
&& \qquad \qquad = \frac{1}{2\pi} \int_{-\infty}^{+\infty} \int_{-\infty}^{+\infty} e^{i r x} [W(\rho)] (x\cos\theta - y\sin\theta , y\cos\theta + x\sin\theta) \de x \de y \\
&& \qquad \qquad = \frac{1}{2\pi} \int_{-\infty}^{+\infty} e^{i r x} [RW(\rho)](\theta , x) \de x \\
&& \qquad \qquad = \frac{1}{\sqrt{2\pi}} \left[ \ff^{-1} [RW(\rho)](\theta , \cdot) \right] (r) .
\end{eqnarray*}
By injectivity of Fourier transform, the claim then follows by comparison.
\end{proof}
\begin{corollary}\label{CorSuppW}
The support of $W(\rho)$ is an unbounded subset of $\R^2$ for all $\rho\in\sh$.
\end{corollary}
\begin{proof}
Suppose by contradiction that $W(\rho) = 0$ almost everywhere outside the disk $D_R$ of radius $R$ in $\R^2$. Then $W(\rho) \in \luno{\R^2}$, and so $[RW(\rho)] (\theta,r) = p^\rho (\theta,r)$ by the above proposition. We have
\begin{eqnarray*}
\int_{0}^{2\pi} \int_{|r| > R} |RW(\rho) (\theta,r)| \de r \frac{\de\theta}{2\pi} & \leq & \int_{0}^{2\pi} \iint_{\R^2 \setminus D_R} |W(\rho) (r\cos \theta - t\sin \theta , r\sin \theta + t\cos \theta)| \de r \de t \frac{\de\theta}{2\pi} \\
& = & \int_{0}^{2\pi} \iint_{\R^2 \setminus D_R} |W(\rho) (r,t)| \de r \de t \frac{\de\theta}{2\pi}  = 0 ,
\end{eqnarray*}
which contradicts Corollary \ref{suppprho}.
\end{proof}

The first formal derivation of \eqref{TdRdW} is contained in Ref.~ \cite{VR}, without the assumption $W(\rho)\in \luno{\R^2}$. We stress that if $W(\rho)\notin \luno{\R^2}$, then \eqref{TdRdW} does not make sense, and the only possible definition of $p^\rho$ is by means of item 1 in Theorem \ref{TeoMis}.

If we denote by $\shuno$ the subset of states $\rho\in\sh$ such that $W(\rho)\in\luno{\R^2}$, then we have estabilished the following diagram
$$
\xymatrix{ & \shuno \ar_W[dl]\ar^{p^{\cdot}}[dr] & \\ \luno{\R^2} \ar@<0.5ex>[rr]^R & & P(\T\times\R)}
$$

Now we turn to the problem of reconstructing $W(\rho)$ given $p^\rho$. If $W(\rho)\in S(\R^2)$, the space of Schwartz functions on $\R^2$, Radon inversion formula is applicable, and we can obtain $W(\rho)$ from $p^\rho$ in a rather explicit way. Before stating Radon inversion theorem, according to Ref.~ \cite{Hel} we need to introduce the set $S_H (\PP^2)$ of functions $\phi : \T \times \R \frecc \C$ such that
\begin{itemize}
\item[({\rm i})] $\phi\in C^\infty (\T\times \R)$;
\item[({\rm ii})] $\sup_{\theta, \, r} \left| \lft 1+|r|^k \rgt
    \frac{\partial^l}{\partial r^l} \frac{\partial^m}{\partial
      \theta^m} \phi(\theta,r) \right| < \infty$; 
\item[({\rm iii})] $\phi(\theta,r) = \phi(2\pi -\theta,-r)$ for all $\theta,r$;
\item[({\rm iv})] for each $k\in\N$, $\int_{-\infty}^{+\infty} \phi(\theta,r) r^k \de r$ is a
  homogeneous polynomial in $\sin\theta$, $\cos\theta$ of degree $k$. 
\end{itemize}
It is shown in Ref.~ \cite{Hel} that $Rf\in S_H (\PP^2)$ if $f\in S(\R^2)$,
and the map $R: S(\R^2) \frecc S_H (\PP^2)$ is one-to-one and
onto. Thus, in our case $W(\rho) \in S(\R^2)$ is equivalent $p^\rho\in
S_H (\PP^2)$ by Proposition \ref{PropTdRdW}. 

The next theorem is a restatement of Theorem 3.6 in Ref.~ \cite{Hel} (see
also Ref.~ \cite{VR} for a formal derivation of \eqref{ATprho}). We stress
that the hypothesis $W(\rho)\in S(\R^2)$ (or, equivalently, $p^\rho\in
S_H (\PP^2)$) is needed in order to give meaning to
\eqref{DefLambda} and to define the integral in \eqref{backpr}. 
\begin{theorem}
Suppose $W(\rho)\in S(\R^2)$. Then
\begin{equation}\label{ATprho}
W(\rho) = \frac{1}{4\pi^2} R^\# [\Lambda p^\rho]
\end{equation}
where
\begin{equation}\label{DefLambda}
\Lambda p^\rho (\theta,r) = \sqrt{\frac{\pi}{2}} \left[ \ff_t [|t|] \ast p^\rho (\theta,\cdot) \right] (r) = {\rm PV} \left[ \int_{-\infty}^{+\infty} \frac{1}{r-t} \frac{\partial p^\rho (\theta,t)}{\partial t} \de t \right]
\end{equation}
and
\begin{equation}\label{backpr}
R^\# f (x,y) = \int_0^{2\pi} f(\theta , x\cos\theta + y\sin\theta) \frac{\de\theta}{2\pi} \quad \forall f\in C^\infty (\T \times \R)
\end{equation}
(in \eqref{DefLambda}, the Fourier transform of $|t|$ and the convolution are interpreted in the sense of tempered distributions, and ${\rm PV}$ is the Cauchy principal value of the integral).
\end{theorem}

We devote the rest of this section to find the subset of states $\rho\in\sh$ such that $W(\rho)\in S(\R^2)$, i.~e.~to which both Radon transform \eqref{TdRdW} and Radon reconstruction formula \eqref{ATprho} are applicable.

Each $T\in\trh$, being a Hilbert-Schmidt operator on $\ldue{\R}$, is an integral operator, whose kernel $K_T$ is in $\ldue{\R^2}$. We have the following fact.
\begin{proposition}
Suppose $K\in S(\R^2)$. Then the integral operator $L_K$ with kernel $K$ is in $\trh$, and its trace is
\begin{equation}\label{trLK}
\tr{L_K} = \int_{-\infty}^{+\infty} K(x,x) \de x .
\end{equation}
Moreover, $L_K \in \sh$ if and only if $K$ is positive semidefinite\footnote{We recall that a function $K:\R^2 \frecc \C$ is {\em positive semidefinite} if $\sum_{i,j = 1}^N c_i \overline{c_j} K(x_j , x_i) \geq 0$ for all $N\in\N$, ${c_1 , c_2 \ldots c_N} \subset \C$ and ${x_1 , x_2 \ldots x_N} \subset \R$.} 
and $\int_{-\infty}^{+\infty} K(x,x) \de x = 1$.
\end{proposition}
\begin{proof}
Let $I = (-\pi , \pi)$, and let $\Phi : \ldue{\R} \frecc \ldue{I}$ be the following unitary operator
$$
\Phi f (y) = (1+\tan^2 y)^{1/2} f (\tan y) .
$$
$\Phi$ intertwines $L_K$ with the integral operator $L_{\tilde{K}}$ on $\ldue{I}$ with kernel
$$
\tilde{K} (y_1 , y_2) = (1+\tan^2 y_1)^{1/2} K (\tan y_1 , \tan y_2) (1+\tan^2 y_2)^{1/2} \quad y_1, y_2 \in (-\pi, \pi) .
$$
Since $\tilde{K}$ extends to a $C^\infty$-function on
$\overline{I\times I}$ by setting $\tilde{K} = 0$ in the frontier of
$\overline{I\times I}$, by Lemma 10.11 in Ref.~ \cite{Knapp} $L_{\tilde{K}}$
is a trace class operator on $\ldue{I}$, whose trace is given by 
$$
\tr{L_{\tilde{K}}} = \int_{-\pi}^{\pi} \tilde{K} (y,y) \de y = \int_{-\infty}^{+\infty} K (x,x) \de x .
$$
Since $L_K = \Phi^{-1} L_{\tilde{K}} \Phi$, eq.~\eqref{trLK} follows. 

It is easy to check that, if $K$ is positive semidefinite, then the integral operator $L_K$ is positive. Conversely, suppose $L_K$ is a positive operator. Fix a Dirac sequence $\{ f_n \}_{n\in\N}$, and let $g_n = \sum_{i=1}^N c_i f_n^{x_i}$, where $f_n^{x_i} (x) = f_n (x-x_i)$. We have
\begin{eqnarray*}
0 & \leq & \scal{g_n}{L_K g_n} = \sum_{i,j=1}^N c_i \overline{c_j} \int_{-\infty}^{+\infty} \int_{-\infty}^{+\infty} \overline{f_n (x-x_j)} K(x,y) f_n (y-x_i) \de x \de y \\ && \mathop{\frecc}_{n\to\infty} \, \sum_{i,j=1}^N c_i \overline{c_j} K(x_j,x_i) ,
\end{eqnarray*}
from which positive definiteness of $K$ follows. The last claim in the statement is thus clear.
\end{proof}

We introduce the following linear subspace of $\trh$
$$
\trhs = \left\{ T\in\trh \mid K_T \in S(\R^2) \right\} ,
$$
and define
$$
\shs = \sh \cap \trhs .
$$
If $T\in\trhs$, we can explicitly evaluate the trace in \eqref{WW} and the Fourier transform in \eqref{W} defining $V(T)$ and $W(T)$ respectively. We find
\begin{gather*}
[V(T)](x,y) = \ff_t^{-1}\left[ K_T ( t + y/2 , t - y/2) \right] (x) \\
[W(T)](x,y) = \frac{1}{\sqrt{2\pi}} \ff_t\left[ K_T ( x + t/2 , x - t/2) \right] (y) ,
\end{gather*}
where we denoted by $\ff_t$ the Fourier transform with respect to the variable $t$.
The second formula proves the next proposition.
\begin{proposition}
$W : \trhs \frecc S(\R^2)$ is a bijection.
\end{proposition}

Restricting to states in $\shs$, we have thus arrived at the following diagram.
$$
\xymatrix{ & \shs \ar_W[dl]\ar^{p^{\cdot}}[dr] & \\ S(\R^2) \ar@<0.5ex>[rr]^R & & S_H (\PP^2) \ar@<0.5ex>[ll]^{\frac{1}{4\pi^2} R^\# \Lambda}}
$$
\begin{remark}\label{RemDistr}
Unfortunately, one can not use the definition of Radon transform of distributions to extend \eqref{TdRdW} to whole $\ldue{\R^2}$, or reconstruction formula \eqref{ATprho} to a larger set than $\shs$. In fact, as explained in \S 5 of Ref.~ \cite{Hel}, the distributional Radon transform can be defined only as a map $R : \ee^\prime (\T \times \R) \frecc \ee^\prime (\T \times \R)$, $\ee^\prime (\T \times \R)$ being the set of compactly supported distributions on $\T \times \R$. Corollary \ref{CorSuppW} then prevents us from giving any distributional sense to \eqref{TdRdW}. Similarly, eq.~\eqref{ATprho} has no distributional analogue, as the reconstruction formula $T = \frac{1}{4\pi^2} R^\# [\Lambda RT]$ (Theorem 5.5 in Ref.~ \cite{Hel}) again holds only for compactly supported distributions $T$.
\end{remark}
\begin{remark}
Being able to exhibit an explicit inversion formula of the Radon transform only for Wigner functions which are Schwartz class does \emph{not} imply 
a failure of quantum tomographical methods in reconstructing states with weaker regularity properties, as the associated POVM remains informationally complete 
on the whole of $\sh$, as we have shown in the first part of this paper. In fact, mainly in order to address issues of numerical stability, actual reconstruction methods usually do not involve 
$\frac{1}{4\pi^2} R^\# \Lambda$ directly, but some approximated technique involving regularizations; proofs of consistency are available \cite{Art} for some of these regularized
estimators which holds on the whole of quantum state space. 
\end{remark}

\end{document}